\documentclass[12pt, a4paper]{article}

\usepackage{doc}

\usepackage{a4wide}

\usepackage{amsmath, amsthm, amssymb}

\usepackage{enumitem}

\usepackage[norefs,nocites]{refcheck}

\newtheorem{counter}{Counter}

\theoremstyle{definition}

\newtheorem{definition}[counter]{Definition}

\theoremstyle{plain}

\newtheorem{theorem}[counter]{Theorem}

\theoremstyle{remark}

\usepackage[plain]{fancyref}

\newcommand*{\fancyrefexlabelprefix}{ex}

\newcommand*{\frefexname}{\text{example}}
\newcommand*{\Frefexname}{\text{Example}}

\frefformat{vario}{\fancyrefexlabelprefix}{\frefexname\fancyrefdefaultspacing#1#3}
\Frefformat{vario}{\fancyrefexlabelprefix}{\Frefexname\fancyrefdefaultspacing#1#3}
\frefformat{plain}{\fancyrefexlabelprefix}{\frefexname\fancyrefdefaultspacing#1}
\Frefformat{plain}{\fancyrefexlabelprefix}{\Frefexname\fancyrefdefaultspacing#1}

\newcommand*{\fancyrefdeflabelprefix}{def}

\newcommand*{\frefdefname}{\text{definition}}
\newcommand*{\Frefdefname}{\text{Definition}}

\frefformat{vario}{\fancyrefdeflabelprefix}{\frefdefname\fancyrefdefaultspacing#1#3} 
\Frefformat{vario}{\fancyrefdeflabelprefix}{\Frefdefname\fancyrefdefaultspacing#1#3}
\frefformat{plain}{\fancyrefdeflabelprefix}{\frefdefname\fancyrefdefaultspacing#1}
\Frefformat{plain}{\fancyrefdeflabelprefix}{\Frefdefname\fancyrefdefaultspacing#1}

\newcommand*{\fancyrefthmlabelprefix}{thm}

\newcommand*{\frefthmname}{\text{theorem}}
\newcommand*{\Frefthmname}{\text{Theorem}}

\frefformat{vario}{\fancyrefthmlabelprefix}{\frefthmname\fancyrefdefaultspacing#1#3}
\Frefformat{vario}{\fancyrefthmlabelprefix}{\Frefthmname\fancyrefdefaultspacing#1#3}
\frefformat{plain}{\fancyrefthmlabelprefix}{\frefthmname\fancyrefdefaultspacing#1}
\Frefformat{plain}{\fancyrefthmlabelprefix}{\Frefthmname\fancyrefdefaultspacing#1}

\newcommand*{\fancyrefremlabelprefix}{rem}

\newcommand*{\frefremname}{\text{remark}}
\newcommand*{\Frefremname}{\text{Remark}}

\frefformat{vario}{\fancyrefremlabelprefix}{\frefremname\fancyrefdefaultspacing#1#3}
\Frefformat{vario}{\fancyrefremlabelprefix}{\Frefremname\fancyrefdefaultspacing#1#3}
\frefformat{plain}{\fancyrefremlabelprefix}{\frefremname\fancyrefdefaultspacing#1}
\Frefformat{plain}{\fancyrefremlabelprefix}{\Frefremname\fancyrefdefaultspacing#1}

\newcommand*{\fancyreflemlabelprefix}{lem}

\newcommand*{\freflemname}{\text{lemma}}
\newcommand*{\Freflemname}{\text{Lemma}}

\frefformat{vario}{\fancyreflemlabelprefix}{\freflemname\fancyrefdefaultspacing#1#3}
\Frefformat{vario}{\fancyreflemlabelprefix}{\Freflemname\fancyrefdefaultspacing#1#3}
\frefformat{plain}{\fancyreflemlabelprefix}{\freflemname\fancyrefdefaultspacing#1}
\Frefformat{plain}{\fancyreflemlabelprefix}{\Freflemname\fancyrefdefaultspacing#1}

\newcommand*{\fancyrefsubseclabelprefix}{subsec}

\newcommand*{\frefsubsecname}{\text{subsection}}
\newcommand*{\Frefsubsecname}{\text{Subsection}}

\frefformat{vario}{\fancyrefdeflabelprefix}{\frefsubsecname\fancyrefdefaultspacing#1#3}
\Frefformat{vario}{\fancyrefsubseclabelprefix}{\Frefsubsecname\fancyrefdefaultspacing#1#3}
\frefformat{plain}{\fancyrefsubseclabelprefix}{\frefsubsecname\fancyrefdefaultspacing#1}
\Frefformat{plain}{\fancyrefsubseclabelprefix}{\Frefsubsecname\fancyrefdefaultspacing#1}

\newcommand*{\fancyrefcorlabelprefix}{cor}

\newcommand*{\frefcorname}{\text{corollary}}
\newcommand*{\Frefcorname}{\text{Corollary}}

\frefformat{vario}{\fancyrefcorlabelprefix}{\frefcorname\fancyrefdefaultspacing#1#3}
\Frefformat{vario}{\fancyrefcorlabelprefix}{\Frefcorname\fancyrefdefaultspacing#1#3}
\frefformat{plain}{\fancyrefcorlabelprefix}{\frefcorname\fancyrefdefaultspacing#1}
\Frefformat{plain}{\fancyrefcorlabelprefix}{\Frefcorname\fancyrefdefaultspacing#1}

\newcommand*{\fancyrefalgolabelprefix}{algo}

\newcommand*{\frefalgoname}{\text{algorithm}}
\newcommand*{\Frefalgoname}{\text{Algorithm}}

\frefformat{vario}{\fancyrefalgolabelprefix}{\frefalgoname\fancyrefdefaultspacing#1#3}
\Frefformat{vario}{\fancyrefalgolabelprefix}{\Frefalgoname\fancyrefdefaultspacing#1#3}
\frefformat{plain}{\fancyrefalgolabelprefix}{\frefalgoname\fancyrefdefaultspacing#1}
\Frefformat{plain}{\fancyrefalgolabelprefix}{\Frefalgoname\fancyrefdefaultspacing#1}

\newcommand*{\fancyrefproplabelprefix}{prop}

\newcommand*{\frefpropname}{\text{proposition}}
\newcommand*{\Frefpropname}{\text{Proposition}}

\frefformat{vario}{\fancyrefproplabelprefix}{\frefpropname\fancyrefdefaultspacing#1#3}
\Frefformat{vario}{\fancyrefproplabelprefix}{\Frefpropname\fancyrefdefaultspacing#1#3}
\frefformat{plain}{\fancyrefproplabelprefix}{\frefpropname\fancyrefdefaultspacing#1}
\Frefformat{plain}{\fancyrefproplabelprefix}{\Frefpropname\fancyrefdefaultspacing#1}

\usepackage{bussproofs}
\newcommand\myLL[1]{\LeftLabel{$^{\text{#1}}$}}

\newcommand\tabProb{\textbf{prob}}

\newcommand\Prop{\mathsf{Prop}}
\newcommand\Nat{\mathbb{N}}
\newcommand\Rat{\mathbb{Q}}
\newcommand\Tm{\mathsf{Tm}}

\newcommand\LLPPOne{\mathcal{L}_{\LPPOne}}
\newcommand\LJ{\mathcal{L}_{\J}}

\newcommand\axP{\mathsf{(P)}}
\newcommand\axJ{\mathsf{(J)}}

\newcommand\ANE{\mathsf{(AN!)}}
\newcommand\PI{\mathsf{(PI)}}
\newcommand\WE{\mathsf{(WE)}}
\newcommand\LE{\mathsf{(LE)}}
\newcommand\DIS{\mathsf{(DIS)}}
\newcommand\UN{\mathsf{(UN)}}

\newcommand\MP{\mathsf{(MP)}}
\newcommand\CE{\mathsf{(CE)}}
\newcommand\ST{\mathsf{(ST)}}

\newcommand\PPJ{\mathsf{PPJ}}

\newcommand\J{\mathsf{J}}

\newcommand\SFour{\mathsf{S4}}
\newcommand\D{\mathsf{D}}
\newcommand\K{\mathsf{K}}

\newcommand\LPPOne{\mathsf{LPP_1}}

\newcommand\LPlogic{\mathsf{LP}}

\newcommand\system{\mathcal{S}}

\newcommand\pspace{\mathsf{PSPACE}}

\newcommand\PPJMeas{{\PPJ}_{\mathsf{, Meas}}}
\newcommand\LPPOneMeas{\mathsf{LPP_{1, Meas}}}

\newcommand\true{\mathsf{T}}
\newcommand\false{\mathsf{F}}

\newcommand\powerset{\mathcal{P}}
\newcommand\evid{\mathcal{E}}
\newcommand\myvec[1]{\boldsymbol{#1}}

\title{The Complexity of 
Probabilistic Justification Logic}

\author{Ioannis Kokkinis\\
LORIA, CNRS-University of Lorraine, 
Nancy, France\\
\texttt{ioannis.kokkinis@loria.fr}}

\begin{document}

\maketitle

\begin{abstract}
Probabilistic justification logic is a modal logic with two
kind of modalities: probability measures
and explicit justification terms.
We present a tableau procedure that can be used
to decide the satisfiability problem for this
logic in polynomial space. We show that this
upper complexity bound is tight.
\end{abstract}

\section{Introduction}

Following \cite{ograma09} we can define
a probabilistic version of a base
logic by enriching the language of the base logic with 
probabilistic operators. The
probabilistic operators create formulas
of the form $P_{\geq s} A$ which
read as ``$A$ holds with probability at least $s$''.
The models
of these probabilistic logics are probability spaces
which have models of the base logic as states.
In order to obtain a sound and complete axiomatization
the usual axioms for probability are combined with
the axioms of the base logic~\cite{ograma09}. 

Artemov developed the first 
justification  logic, the Logic of Proofs
($\LPlogic$), 
to provide intuitionistic logic with a 
classical provability
semantics~\cite{Art95TR,Art01BSL}.
In \cite{Art01BSL}
it was proved that any theorem of modal
logic $\SFour$ can be
translated into a theorem of
$\LPlogic$ by replacing any occurrence
of the modal operator $\Box$ with an
appropriate explicit justification term
and that any theorem in $\LPlogic$
can be translated into a theorem in $\SFour$ by
replacing any occurrence of a
justification term with a $\Box$.
In the same way explicit
counterparts for several modal logics were
found~\cite{Bre00TR}.
For example the justification logic $\J$
is the explicit counterpart of the minimal
modal logic $\K$.

In~\cite{koogst} a probabilistic
justification logic, $\PPJ$, is defined over the basic
justification logic $\J$~\cite{Bre00TR}.
In this paper we present a tableau procedure
that can be used to decide the satisfiability
problem in $\PPJ$.
This procedure uses a rule that is applied to all
the formulas that appear in the scope of
some probabilistic operator in a tableau branch.
The rule creates exponentially many branches,
however by applying a theorem from the theory of linear systems
we show that only polynomially many branches are needed in order
to decide the satisfiability of a given formula.
This way we can decide the satisfiability problem for
$\PPJ$ in polynomial
space. We show that our upper bound is tight
via a reduction from modal 
logic $\D$, which is the modal logic
that is complete with respect to serial Kripke
structures.

\section{A Probabilistic Logic over Classical Propositional Logic}

Let $\Prop$ be a countable set of atomic
propositions. The logic $\LPPOne$ is defined in
\cite{ograma09} over the language $\LLPPOne$:
\[
A ::=  p ~|~ P_{\geq s} A ~|~ \lnot A ~|
~ A \land A
\]
where\footnote{$\Rat$ denotes the set of
rational numbers.}
$s \in \Rat \cap [0,1]$ and $p \in \Prop$. We also use the following abbreviations:
\begin{align*}
P_{< s} A \equiv \lnot P_{\geq s} A~,~
P_{\leq s} A \equiv  P_{\geq 1 - s} \lnot 
A ~,~
P_{> s} A \equiv \lnot P_{\leq s} 
A \text{ and }
P_{= s} A \equiv P_{\geq s} A \land 
P_{\leq s} A~.
\end{align*}

The axiom schemata and the derivation rules
of the logic $\LPPOne$ are presented in \Fref{tab:LPP1}.
Axiom $\PI$ corresponds to the fact that the
probability of truthfulness of every
formula is at least $0$. 
Axioms $\WE$ and $\LE$ describe some properties of
inequalities.
Axioms $\DIS$ and $\UN$ correspond to the additivity of
probabilities for disjoint events.
The rule $\CE$ is the probabilistic analogue
of the modal necessitation rule and
the rule $\ST$ informally says that if the probability
of a formula is arbitrarily close to $s$ then it
is at least $s$. $\ST$ corresponds to the Archimedean
property of the real numbers.

\begin{table}[ht]
\centering
\renewcommand*{\arraystretch}{1.25}
\begin{tabular}{|c l|}
\hline 
& Axiom Schemata: \\
$\axP$ & finitely many axioms schemata for classical propositional logic\\
$\PI$ & $\vdash P_{\geq 0} A$\\
$\WE$ & $\vdash P_{\leq r} A \to P_{< s} A$, where $s > r$\\
$\LE$ & $\vdash P_{< s} A \to P_{\leq s} A$\\
$\DIS$ & $\vdash  P_{\geq r} A \land P_{\geq s} B \land P_{\geq 1} \lnot (A \land B) \to P_{\geq \min(1, r+s)} (A \lor B)$\\
$\UN$ & $\vdash P_{\leq r} A \land P_{< s} B \to P_{<r+s} (A \lor B)$, where $r+s \leq 1$\\
& Derivation Rules:\\
$\MP$ & if $T \vdash A$ and $T \vdash A \to B$ then
$T \vdash B$\\
$\CE$ & if $\vdash A$  then $\vdash P_{\geq 1 } A$\\
$\ST$ & if $T \vdash A \to P_{\geq s - \frac{1}{k} } B$
for every integer $k \geq \frac{1}{s}$ and $s > 0$
then $T \vdash A \to P_{\geq s } B$\\
\hline
\end{tabular}
\caption{\label{tab:LPP1} Axiom Schemata
and Derivation Rules of $\LPPOne$}
\end{table}

A probability space is a triple
$\langle W, H, \allowbreak \mu\rangle$,
where $W$ is a non-empty set of states,
$H \subseteq \powerset(W)$ 
($\powerset$ stands for powerset)
is closed under finite union and complementation 
and $\mu : H \to [0,1]$ such that
$\mu(W) = 1$ and for any disjoint $U$ and $V$
in $H$, $\mu(U \cup V) = \mu (U) + \mu(V)$.
The models for $\LPPOne$ are probability
spaces where the states contain 
truth assignments
and probability
spaces (so that we can deal with iterated probabilities).

\begin{definition}[$\LPPOne$-Model]
\label{def:LPP1model}
An \emph{$\LPPOne$-model} is a quintuple
$M = \langle U, W, H, \mu , v \rangle$ where:
\begin{enumerate}
\item
$U$ is a non-empty set of objects called worlds;
\item
$W, H, \mu$ and $v$ are functions, which have $U$
as their domain,
such that for every $w \in U$:
\begin{center}
$\langle W_w, H_w, \mu_w\rangle$ is
a probability space with $W_w \subseteq U$
and $v_w : \Prop \to \{ \true, \false\}$,
where $\true$ ($\false$) stand for true (false).
\end{center}
\end{enumerate}
\end{definition}

\begin{definition}[Satisfiability in an $\LPPOne$-model]
\label{def:LPPOne_meas_sat}
Let $M = \langle U, W, H, \mu,v \rangle$ be an
$\LPPOne$-model. Satisfiability is defined
as follows (the propositional cases are treated
classically):
\begin{align*}
M, w \models p &\quad\Longleftrightarrow\quad v_w (p) = \true \quad\text{ for $p \in \Prop$}~;\\
M, w \models P_{\geq s} B &\quad \Longleftrightarrow\quad
\Big ( \mu_w \big ( [A]_{M,w} \big ) \geq s \Big ),
\text{ where } 
[A]_{M,w} = \{ u \in W_w ~|~ M,u \models A \}~.
\end{align*}
\end{definition}

Let $M=\langle U, W, H, \mu , v \rangle$
be an $\LPPOne$-model.
$M$ will be called measurable if
for every $w \in U$ and for every $A \in 
\LLPPOne$, $[A]_{M,w} \in H_w$.
In the rest of the paper we restrict ourselves
to measurable models. $\LPPOneMeas$ denotes the class 
of $\LPPOne$-measurable models.

Soundness and strong completeness for $\LPPOne$ with
respect to $\LPPOneMeas$ is proved in~\cite{ograma09}.
Assume that $A_1, \ldots, A_k$ are the
subformulas of some
$A \in \LLPPOne$. A formula of the form
$\pm A_1 \land \ldots \land \pm A_k$,
where $\pm A_i$ is either
$A_i$ or $\lnot A_i$, will be called an atom of $A$. 
In an atom the order of the conjuncts does not matter. So, 
two atoms are considered the same if they have the same 
conjuncts.
$|A|$ is defined as the number of
symbols that are used in order to write $A$ (where all
rational numbers are assumed to have size $1$).
For $A \in \LLPPOne$, $||A||$ is
the biggest size of a rational number that appears in $A$ (where the size of a rational number is equal to the sum of the lengths
of the binary representations of its numerator and denominator, when the rational number
is written as an irreducible fraction).

As we mentioned in the introduction,
a well known theorem
from the theory of linear systems is necessary for our results. We present
this theorem as \Fref{thm:linear_system_thm}.
This result is stated (and proved) for the purposes of probabilistic logic as Theorem 5.1.5. in \cite{kok16phd}.
The interesting part of
\Fref{thm:linear_system_thm} is
proved in~\cite[p. 145]{chvatal83}.

\begin{theorem}
\label{thm:linear_system_thm}
Let $\system$ be a linear system of
$n$ variables and of
$r$ linear equalities and/or 
inequalities with integer coefficients
each of size at most $l$.
Assume that the vector \mbox{$\myvec{x} = x_1, \ldots, x_n$}
is a solution of $\system$ such that for all $i \in 
\{ 1, \ldots, n\}$, $x_i \geq 0$.
Then there is a vector $\myvec{x^*} = x^*_1, \ldots, x^*_n$
with the following properties:
\begin{enumerate}[label=(\arabic*), topsep = 0.3em]
\item
$\myvec{x^*}$ is a solution of $\system$ and
at most $r$ entries of $\myvec{x^*}$ are positive;
\item
for all $i$, $x^*_i$ is a non-negative
rational number with size bounded 
by
$2 \cdot \big ( r \cdot l+ r \cdot \log_2 (r) + 1 
\big )$;
\item
for all $i \in \{ 1, \ldots, n\}$, if $x^*_i > 0$ then $x_i > 0$.
\end{enumerate}
\end{theorem}

Now we can prove the small model property for
$\LPPOne$.

\begin{theorem}[Small Model Property for
$\LPPOne$]
\label{thm:smpLLP1}
If $A$ is $\LPPOneMeas$-satisfiable then it
is satisfiable in a model
$M = \langle U, W, H, \mu, v \rangle$
that satisfies the following properties:
\begin{enumerate}[label=(\arabic*)]
\item
$|U| \leq 2^{|A|}$ and
in every world of $U$ exactly one atom of $A$ holds.
\item
For every $w \in U$ the following holds:
\begin{enumerate}
\item
$W_w = U$ and $H_w$ is the powerset of $U$~.
\item
For every $u \in W_w$,
$\mu_w(\{ u \}) \leq
2 \cdot \big ( |A| \cdot || A || + |A| \cdot \log_2 (|A|) + 1\big )$ and $\mu_w(\{ u \}) \in \Rat$.
\item
For every $V \in H_w$:
$\mu_w (V) = \sum_{u \in V} \mu_w(\{ u\})$~.
\item
\label{enum:polynom_pos_prob}
The number of $u$'s
such that $\mu_w(\{ u\})>0$, is at most $|A|$.
\end{enumerate}
\end{enumerate}
\end{theorem}

\begin{proof}
In the proof of Lemma 5.3.6 of \cite{kok16phd}
a model for $A$ that satisfies the conditions of the 
theorem is constructed.
The most interesting property of the small model
is \ref{enum:polynom_pos_prob}, which can be
proved by an application of
\Fref{thm:linear_system_thm}.
\end{proof}

\begin{theorem}
\label{thm:LPP1_lower_bound}
The $\LPPOneMeas$-satisfiability
problem is $\pspace$-hard.
\end{theorem}

\begin{proof}
Since probability spaces are non-empty
sets it makes sense to draw
a reduction from modal logic $\D$,
which is complete for serial Kripke structures.
Let $A$ be a modal formula and let $f(A)$
be the $\LLPPOne$-formula that is obtained by
replacing any occurrence of $\Box$ in $A$
with $P_{\geq 1}$. We will prove that
$A$ is satisfiable iff $f(A)$ is
$\LPPOneMeas$-satisfiable.

Assume that $A$ is satisfiable. Then $A$
is satisfiable in a finite model~\cite{halmos92}.
We can create an $\LPPOneMeas$-model
where the probability space of each world $w$
consists of the worlds accessible to $w$
and $w$ assigns a uniform probability
to each of these worlds. Then we can prove
that $f(A)$ is satisfied in this
$\LPPOneMeas$-model.

Assume that $f(A)$ is satisfiable. Then it
is satisfiable in a model that has the
properties of \Fref{thm:smpLLP1}. We define
a Kripke model where $u$ is accessible from $w$,
if $\mu_w(\{ u\}) > 0$. Then we can prove
that $A$ is satisfiable in this Kripke model.
\end{proof}

\section{Adding Justifications}

Justification logics are modal logics
that use explicit terms instead of the modality
$\Box$. The terms are constructed according to
the grammar
$t :: = c ~|~ x ~|~ (t \cdot t) ~|~!t$
where $c$ is a constant and $x$ is a variable. 
$\Tm$ denotes the set of all terms.
For $t \in \Tm$ and any non-negative integer~$n$
we define:
$!^0 t := t$
and
$!^{n+1}t := {!} ~ ({!^n} t)$.
The language of justification logic, $\LJ$,
is defined by the grammar
$A :: = p ~|~ \lnot A ~|~ A \land A ~|~ t: A$
where $t \in \Tm$ and $p \in \Prop$. 
For this presentation we take
$\axJ$, i.e. $\vdash u : (A \to B) \to
( v : A \to u \cdot v : B )$
as the 
only axiom of the logic $\J$.
The logic $\J$ is defined\footnote{This paper
aims at illustrating the combination
of justification logic and probabilistic 
logic. Therefore, we consider it
useful to study the smallest possible
framework. As a consequence we present
a variant of logic $\J$ without the operator
$+$ and with the maximal constant
specification. Other features of
justification logic, like the term
operator $+$, other justification axioms etc.
can be added to our framework without
complications.} by taking
a system for classical propositional logic, 
the axiom $\J$ and the rule $\ANE$:
\[
\vdash {!^{n}} c : {!^{n-1}} c : \cdots : {!c} : c : A,
\text{ where $c$ is a constant, $A$
is an axiom-instance and } n \in \Nat~.
\]

Semantics for $\J$ are given by $M$-models.

\begin{definition}[$M$-model]
An $M$-model is a pair $\langle v, \evid
\rangle$, where $v: \Prop \to \{\true, \false\}$
and $\evid: \Tm \to \powerset(\LJ)$
such that for every
$u, v \in \Tm$, for a constant $c$ and $A \in \LJ$ 
we have:
\begin{enumerate}
\item
$\big( A \to B \in \evid (u) \text{ and }
A \in \evid (v) \big)
\Longrightarrow B \in \evid (u \cdot v)$~;
\item
if $c$ is a constant, $A$ an axiom
and $n \in \Nat$ then 
${!^{n-1} c} : {!^{n-2}} c: \cdots : !c : c : A \in 
\evid (!^n c)$.
\end{enumerate}
\end{definition}

The language of probabilistic justification
logic is defined as a combination
of $\LJ$ and $\LLPPOne$:
$A :: = p ~|~ \lnot A ~|~ A \land A ~|~ t: A
~|~ P_{\geq s} A$.
So $\PPJ$~\cite{koogst} is
defined by taking the axioms and rules
of $\LPPOne$ together with the axiom and rules
of $\J$ ($\ANE$ can now be applied
to probabilistic axioms instances too). 
A measurable model for probabilistic justification
logic is defined by replacing the truth
assignment in~\Fref{def:LPP1model} with an $M$-model.
The class of measurable models is $\PPJMeas$.
Satisfiability in $\PPJMeas$ is defined
by adding the line
$M, w \models t : A \Longleftrightarrow
A \in \evid_w (t)$
in \Fref{def:LPPOne_meas_sat}, where
$\evid_w$ is the evidence function that
corresponds to the M-model assigned to the world
$w$. Soundness and completeness of probabilistic
logics with respect to
measurable models is proved in~\cite{kok16phd}.
\Fref{thm:smpLLP1} holds for $\PPJ$ as
well~\cite{kok16phd}.

\section{The Tableau Procedure}

Our tableaux are trees where the nodes  are
formulas prefixed with world and truth signs.
So, the node $w~T~A$ ($w~F~A$) intuitively
means that formula $A$ is true 
(resp. false) at world $w$.
A branch is a path
that starts at a result of
an application of the rule $\tabProb$ (defined later)
or at the root and ends at the premise of
an application of the rule $\tabProb$
or at a leaf.
A branch is called closed if it contains
both $w~T~A$ and $w~F~A$ for some $A$. Otherwise it is called
open. A branch is called complete if no rule is applicable
in this branch. Otherwise it is called incomplete.
The only rule
that can create new worlds in our
tableaux is the rule $\tabProb$.
For this reason we can assign
a world to each branch (of course the
same world may be assigned to several
branches). So, $b_w$ denotes a
branch where all the formulas
are prefixed with $w$.
We will use the abbreviation ``$w~T~A$''
(``$w~F~A$'') to denote that
the node $w~T~A$ ($w~T~A$) appears in the 
tableau.
Our tableau rules are the rules for
classical propositional logic and the rule
$\tabProb$: 

\begin{prooftree}
\AxiomC{$b_w$}
\myLL{\tabProb}
\UnaryInfC{$w.1 ~ T~ \pm B_{11} \land \ldots \land \pm B_{1m} ~|~ \cdots ~|~ w.n ~ T~ \pm B_{n1} \land \ldots \land \pm B_{nm} $}
\end{prooftree}

In rule $\tabProb$ the $w.i$'s are new world prefixes and for all $i, j$, ``$w~T~P_{\geq s_{ij}} B_{ij}$'' or ``$w~F~P_{\geq s_{ij}} B_{ij}$''. In our
tableaux we treat formulas starting with
a justification term as atomic formulas.
In other words, no rule can be applied to a formula
of the form $t:A$. 
The tableau procedure consists of two parts:
first we apply the rules and then we mark worlds
and applications of $\tabProb$ satisfiable. 
Assume that $A$ is a formula that we want to test
for satisfiability. We take
$w~T~A$ as the root of the tableau and
then we apply the following steps:
\begin{enumerate}
\item
\label{enum:tab_rule_prop}
Apply the propositional rules for
as long as possible.
If there exists an open branch that
contains $w~T~P_{\geq s} B$ or
$w~F~P_{\geq s} B$ for some $s$ and $B$ 
then we go to
step \ref{enum:tab_rule_prob}.
Otherwise stop.
\item
\label{enum:tab_rule_prob}
Apply the probabilistic tableau rule to
every open branch. Then go to 
step~\ref{enum:tab_rule_prop}.
\end{enumerate}

The second part of the tableau procedure consists
of a method for marking worlds and
applications of $\tabProb$ satisfiable.
In order to mark
worlds satisfiable we traverse the tree from
the leaves to the root and we make sure that
the justification and the probabilistic 
restrictions are satisfied. In order to check
the probabilistic constraints we have
to mark applications of {\tabProb} as satisfiable
as well.

\textbf{Marking Worlds Satisfiable.}
Let $b_w$ be one of the branches that
correspond to world $w$. In order
to check that ``justification constraints''
hold in $b_w$
we have to extend the satisfiability algorithm for
justification logic $\J$~\cite{Kuz00CSLnonote}
in the probabilistic context.
The algorithm of \cite{Kuz00CSLnonote}
checks that if ``$w~F~t:A$''
then $A \notin \evid(t)$ where
$\langle v, \evid \rangle$ is the minimum
$M$-model that is defined by
the formulas $u:B$ such that ``$w~T~u:B$''.
This algorithm 
uses a procedure for unifying
axiom schemata of justification logic. In order
to extend this algorithm to the probabilistic
setting we have to extend the unification
to probabilistic axiom schemata.
These axioms come with some linear side
conditions~(see~\Fref{tab:LPP1}), so their
unification will create a linear system.
The unification algorithm then succeeds if this
linear system is satisfiable.
For more details  see Lemma 5.3.3 of~\cite{kok16phd}.
Now $w$ will be marked satisfiable if
there exists an open $b_w$
such that the extended algorithm
for justification satisfiability holds
and either it
is a complete branch or it ends in
an application of {\tabProb} and this
application is marked satisfiable.

\textbf{Marking Applications of {\tabProb}
Satisfiable.}
Let $\rho$ be an application of the rule {\tabProb}
on branch $b_w$.
We associate variables $x_i$ with every
world $w.i$, even if $w_i$ is marked not
satisfiable. The $x_i$ corresponds to the probabilities
that world $w$ assigns to $w_i$ in
a small model for $A$ (i.e. $x_i = \mu_w(\{ w.i\}
)$ in the sense of \Fref{thm:smpLLP1}).
We mark $\rho$ satisfiable
if the following linear system is solvable:

\begin{align*}
\sum^{n}_{i=1}  & x_{i} = 1\\
(\forall 1 \leq i \leq n )  & \big [ x_{i} \geq 0 \big ]\\
\text{if  ``$w ~T~P_{\geq s} C$'' then}  &
\sum_{\{i | \text{``$w.i ~T~ C$''}\}} x_{i} \geq s\\
\text{if  ``$w ~F~P_{\geq s} C$'' then} & \sum_{\{i | \text{``$w.i ~T~ C$''} \} } x_{i} < s~.
\end{align*}

If the initial formula $A$ belongs to a world that is
marked satisfiable then we return
satisfiable. After an application
of a propositional rule the
length of the formula decreases. After
an application of the probabilistic rule
the nesting depth of probabilistic operators decreases.
Hence, our tableau procedure terminates.
By the procedure of marking worlds
satisfiable and by \Fref{thm:smpLLP1}
we get the following theorem.

\begin{theorem}
Let $A$ be a $\PPJ$-formula. 
The tableau method
returns $A$ is satisfiable iff
$A$ is satisfiable in a measurable model.
\end{theorem}

\begin{theorem}
The satisfiability problem for $\PPJ$ is
$\pspace$-complete.
\end{theorem}

\begin{proof}
The lower bound follows from
\Fref{thm:LPP1_lower_bound}. The upper bound
follows by the fact that we can traverse the
tableau for probabilistic justification logic
in a depth first fashion by reusing space.
Whether a world $w$ will be
marked satisfiable depends on the worlds
that appear below it in the tableau.
We only need a polynomial number of bits
that can be reused in order to
decide the satisfiability of the linear
systems and 
justification constraints. A complication
arises since rule {\tabProb} creates
exponentially many worlds. However, because of
\Fref{thm:smpLLP1}\ref{enum:polynom_pos_prob} in every
application of rule $\tabProb$ we can guess
a linear number of branches to which
we will assign non-zero probability.
We conclude that the depth first search operates
in non deterministic polynomial space.
\end{proof}

\noindent \textbf{Acknowledgements:}
The author is grateful to Antonis 
Achilleos and the anonymous referees
for many useful comments and to ERC
for financial support (project EPS $313360$).

\bibliographystyle{plain}

\end{document}